\documentclass[12pt, a4paper, reqno]{amsart}
\usepackage{mathrsfs}
\usepackage{bbm}
\usepackage{amsmath,amscd,amssymb,latexsym}
\usepackage[all]{xy}

\input xypic

\evensidemargin=\oddsidemargin
\DeclareMathVersion{can}
\DeclareMathAlphabet{\can}{OT1}{cmss}{m}{n}
\newtheorem{thm}{Theorem}[section]
\newtheorem{cor}[thm]{Corollary}
\newtheorem{lem}[thm]{Lemma}

\newtheorem{exa}[thm]{Example}
\theoremstyle{definition}

\theoremstyle{fact}

\theoremstyle{conjecture}

\numberwithin{equation}{section}

\newcommand{\ord}{\operatorname{ord}}

\begin{document}
\title[Cyclic codes]{Weight distributions of  cyclic codes with respect to pairwise coprime order elements}

\author[C. Li] {Chengju Li}
\address{\rm Department of Mathematics, Nanjing University of Aeronautics and Astronautics,
Nanjing, 211100, P.R. China
} \email{lichengju1987@163.com}
\author[Q. Yue]{Qin Yue}
\address{\rm Department of Mathematics, Nanjing University of Aeronautics and Astronautics,
Nanjing, 211100, P. R. China} \email{yueqin@nuaa.edu.cn}
\author[F. Li] {Fengwei Li}
\address{\rm School of Mathematics and Statistics, Zaozhuang University, Zaozhuang,
277160, P. R. China \\ Department of Mathematics, Nanjing University of Aeronautics and Astronautics,
Nanjing, 211100, P.R. China}\email{lfwzzu@126.com}
\thanks{The paper is supported by NNSF of China (No. 11171150)}

\begin{abstract}
 Let $\Bbb F_r$ be an extension of a finite field $\Bbb F_q$ with  $r=q^m$.
Let each $g_i$ be of order $n_i$ in $\Bbb F_r^*$ and $\gcd(n_i, n_j)=1$ for $1\leq i \neq j  \leq u$.
 We define a cyclic code over $\Bbb F_q$ by
 $$\mathcal C_{(q, m, n_1,n_2, \ldots, n_u)}=\{ c(a_1, a_2, \ldots, a_u) : a_1, a_2, \ldots, a_u \in \Bbb F_r\},$$
 where
 $$c(a_1, a_2, \ldots, a_u)=(\mbox{Tr}_{r/q}(\sum_{i=1}^ua_ig_i^0), \ldots, \mbox{Tr}_{r/q}(\sum_{i=1}^ua_ig_i^{n-1} ))$$
and $n=n_1n_2\cdots n_u$. In this paper, we present a method to compute the weights of
$\mathcal C_{(q, m, n_1,n_2, \ldots, n_u)}$. Further, we  determine the weight distributions of
the cyclic codes $\mathcal C_{(q, m, n_1,n_2)}$ and $\mathcal C_{(q, m, n_1,n_2,1)}$.

\end{abstract}
\maketitle

\section{Introduction}

Let $\Bbb F_q$ be a finite field with $q$ elements, where $q=p^s$, $p$ is a prime, and
$s$ is a positive integer. An $[n, k, d]$ linear code $\mathcal{C}$ is a $k$-dimensional
subspace of $\Bbb F_q^n$ with minimum distance $d$. It is called cyclic if
$(c_0, c_1, \ldots, c_{n-1}) \in \mathcal C$ implies
$(c_{n-1}, c_0, c_1, \ldots, c_{n-2}) \in \mathcal C$.
By identifying the vector $(c_0, c_1, \ldots, c_{n-1}) \in \Bbb F_q^n$ with
$$c_0+c_1x+c_2x^2+\cdots+c_{n-1}c^{n-1} \in \Bbb F_q[x]/(x^n-1),$$
any code $\mathcal C$ of length $n$ over $\Bbb F_q$ corresponds to a subset of
$\Bbb F_q[x]/(x^n-1)$. Then $\mathcal C$ is a cyclic code if and only if the
corresponding subset is an ideal of $\Bbb F_q[x]/(x^n-1)$. Note that every
ideal of $\Bbb F_q[x]/(x^n-1)$ is principal. Hence there is a monic polynomial $g(x)$ with least
degree such that $\mathcal C=\langle g(x) \rangle$ and $g(x) \mid (x^n-1)$. Then $g(x)$ is called the generator
polynomial and $h(x)=(x^n-1)/g(x)$ is called the parity-check polynomial of the cyclic code
$\mathcal C$. Suppose that $h(x)$ has $u$ irreducible factors over $\Bbb F_q$, we call
$\mathcal C$ the dual of the cyclic code with $u$ zeros.

 Let $A_i$ be the number of codewords
with Hamming weight $i$ in the code $\mathcal C$ of length $n$. The weight enumerator of
$\mathcal C$ is defined by
$$1+A_1x+A_2x^2+\cdots+A_nx^n.$$
The sequence $(1, A_1, A_2, \ldots, A_n)$ is called the weight distribution of the code
$\mathcal C$.  In coding theory it is often desirable to know the weight distributions
of the codes because they can be used to estimate the error correcting capability and
the error probability of error detection and correction with respect to some algorithms.
This is quite useful in practice. Unfortunately, it is a very hard problem in general
and remains open for most cyclic codes.

 Let $r=q^m$ for a positive integer $m$ and $\alpha$ a generator of $\Bbb F_r^\ast$.
Let $h(x)=h_1(x)h_2(x) \cdots h_u(x)$, where $h_i(x) (1 \leq i \leq u)$ are distinct monic irreducible
polynomials over $\Bbb F_q$.  Let $g_i^{-1}=\alpha^{-N_i}$  be a root of $h_i(x)$ and $n_i$ the order of $g_i$
for $1 \leq i \leq u$. Denote $\delta=\gcd(r-1, N_1, N_2, \ldots, N_u)$ and $n=\frac {r-1} \delta$.
The cyclic code $\mathcal C$ can be defined by
\begin{equation}\mathcal C_{(q, m, n_1, \ldots, n_u)}=\{ c(a_1, a_2, \ldots, a_u) : a_1, a_2, \ldots, a_u \in \Bbb F_r\},\end{equation}
where \begin{equation}c(a_1, a_2, \ldots, a_u)=(\mbox{Tr}_{r/q}(\sum_{i=1}^ua_ig_i^0), \ldots, \mbox{Tr}_{r/q}(\sum_{i=1}^ua_ig_i^{n-1} ))\end{equation}and
$\mbox{Tr}_{r/q}$ denotes the trace function from $\Bbb F_r$ to $\Bbb F_q$.
It follows from Delsarte's Theorem \cite{D} that the code $\mathcal C$ is an $[n, k]$ cyclic code
over $\Bbb F_q$ with the parity-check polynomial $h(x)$, where $k=\deg (h_1(x))+\deg (h_2(x))+ \cdots +
\deg (h_u(x))$.  The weight distributions of such cyclic codes have been studied for many years and are known
in some cases. We describe the known results as follows.
 \begin{enumerate}
   \item For $u=1$, $\mathcal C_{(q, m, n_1)}$ is called an irreducible cyclic code. The weight distributions
   of irreducible cyclic codes have been extensively studied and can be found in
   \cite{BM72, BMY, DG, D09, DY, MC1, MC2}.
   \item  For $u=2$, i.e., $h(x)=h_1(x)h_2(x)$. The duals of the cyclic codes with two zeros have been
   well investigated when $\deg (h_1(x))=\deg (h_2(x))$.  If $g_1$ and
   $g_2$ have the same order in $\Bbb F_r^\ast$, we know $\deg (h_1(x))=\deg (h_2(x))$. Then the weight distribution of
      cyclic code $\mathcal C_{(q,m,n_1,n_1)}$ had been determined for some special cases \cite{DMLZ, FL, FM, Ma, V, W, X1, X2, X3, ZD}.
   If $\Bbb F_r^\ast=\langle g_1 \rangle=\langle g_2 \rangle$, then the weight distribution of
   the code $\mathcal C_{(q,m,r-1,r-1)}$ which is called the dual of primitive cyclic code with two zeros had been studied
   in \cite{BM, CCD, CCZ, C, F, LTW, Mc, M, S, YCD}.
   \item For $u=3$.  The results on the weight distribution of the dual of cyclic code $\mathcal C_{(q,m,n_1,n_2,n_3)}$  with three zeros can be found
   in \cite{LF, Z, ZDLZ} where the orders are not pairwise coprime.
   \item For arbitrary $u$. Yang et al. \cite{YXD} described a class of the duals of cyclic codes with $u$ zeros and determined their weight
   distributions under special conditions. Li et al. \cite{LHFG} also studied such cyclic codes  and developed a connection
   between the weight distribution and the spectra of Hermitian forms graphs.
 \end{enumerate}

   In the following, we always assume that $r-1=n_iN_i$ for $1 \leq i \leq u$  and $\gcd(n_i, n_j)=1$ for $i \neq j$.
Let $\alpha$ be a generator of $\Bbb F_r^\ast$, where $r=q^m$ for a positive integer $m$.
For $g_i=\alpha^{N_i}$, it is easily known that $g_i$ and $g_j$ are not conjugates of each other
 over $\Bbb F_q$ when $i \neq j$.  Denote $\delta=\gcd(r-1, N_1, N_2, \ldots, N_u)$, and $n=\frac {r-1} \delta$.
In fact, we have $n=n_1n_2 \cdots n_u$ because
$$\delta n_1 n_2 \cdots n_u=\gcd((r-1)n_1n_2 \cdots n_u, (r-1)n_2 \cdots n_u, \ldots, (r-1)n_1n_2 \cdots n_{u-1})=r-1.$$
In this paper, we present a method to compute the weight distribution of
$\mathcal C_{(q, r, n_1,n_2, \ldots, n_u)}$ when $\gcd(n_i, n_j)=1$ for $i \neq j$.
Further, we  determine the weight distributions of
the cyclic codes $\mathcal C_{(q, m, n_1,n_2)}$ and $\mathcal C_{(q, m, n_1,n_2,1)}$ clearly
by using the Gauss periods.

 This paper is organized as follows. In Section 2, we introduce some results about Gauss periods.
 In Section 3, we shall present a method to compute the weights of the
  cyclic code $\mathcal C_{(q, m, n_1,n_2, \ldots, n_u)}$. In Section 4, we
   obtain the weight distributions of
the cyclic code $\mathcal C_{(q, m, n_1,n_2)}$ by using Gauss periods and some examples.
In Section 5, we give the weight distributions of the cyclic code  $\mathcal C_{(q, m, n_1,n_2,1)}$
by using Gauss periods and  some examples.

For convenience, we introduce the following notation in this paper:

\begin{tabular}{ll}
$\Bbb F_r$ & finite field of $r$ elements, $r=q^m$, \\
$r-1=n_1N_1=n_2N_2$ & integer factorization of $r-1$,\\
$\alpha, \alpha^{\frac {r-1} {q-1}}$ & generators of $\Bbb F_r^\ast$ and $\Bbb F_q^\ast$, \\
$C_0^{(N_1,r)}=\langle\alpha^{N_1}\rangle$, & subgroup of order $\frac{r-1}{N_1}$ of $\Bbb F_r^\ast$, \\
$C_i^{(N_1, r)}=\alpha^iC_0^{(N_1,r)}$ & $i$th cyclotomic class of order $N_1$ in $\Bbb F_r^\ast$,\\
$\eta_i^{(N_1, r)}$ &  Gauss period of order $N_1$ over $\Bbb F_r^\ast$, \\
$\mbox{Tr}_{r/q}$ & trace function from $\Bbb F_r$ to $\Bbb F_q$, \\
$\phi$ &  canonical additive character of $\Bbb F_q$, \\
$\psi$ &   canonical additive character of $\Bbb F_r$, \\
$N_0=\frac{r-1} {q-1}, \Bbb F_q^\ast=C_0^{(N_0, r)}$, \\
$d_1=\gcd({N_0, N_1})$, \\
$ d_2=\gcd({N_0, N_2})$, \\
$ d=\gcd({\frac{N_0N_2} {d_2}, N_1})$. \\
\end{tabular}

\section{Gauss periods}

Let ${\Bbb F}_r$ be the finite field with $r$ elements, where $r$ is a power of prime $p$.
For any $a \in {\Bbb F}_r$, we can define
an additive character of the finite field ${\Bbb F}_r$ as follows:
$$\psi_a : {\Bbb F}_r \rightarrow {\Bbb C}^\ast, \psi_a(x)=\zeta_p^{\mbox{Tr}_{r/p}(ax)},$$
where $\zeta_p=e^{\frac{2\pi i}p}$ is a $p$-th primitive  root of unity and $\mbox{Tr}_{r/p}$ denotes
the trace function from ${\Bbb F}_r$ to ${\Bbb F}_p$. If $a=1$, then $\psi_1$ is
called the canonical additive character of $\Bbb F_r$. The orthogonal property of additive characters which can
be found in \cite{LN} is given by
$$\sum_{x \in \Bbb F_r} \psi_1(ax)=\left\{\begin{array}{ll}
r, \mbox{ if } a=0; \\
0, \mbox{ if } a \in {\Bbb F}_r ^\ast. \end{array}\right.$$

Let $r-1=nN$ and $\alpha$ a fixed primitive element of $\Bbb F_r$, where $r=q^m=p^{sm}$.
We define $C_i^{(N, r)}=\alpha^i \langle \alpha^N \rangle$ for $i=0, 1, \ldots, N-1$,
where $\langle \alpha^N \rangle$ denotes the subgroup of $\Bbb F_r^\ast$ generated
by $\alpha^N$. The Gauss periods of order $N$ are given by
$$\eta_i^{(N, r)}=\sum_{x \in C_i^{(N, r)}}\psi(x),$$
where $\psi$ is the canonical additive character of $\Bbb F_r$ and $\eta_i^{(N, r)}=\eta_{i \pmod N}^{(N, r)}$ if $i \geq N$.
In general, the explicit evaluation of Gauss periods is a very difficult problem.
 However, they can be computed in a few cases.

\begin{lem} \cite{My} When $N=2$, the Gauss periods are given by
$$\eta_0^{(2, r)}=\left\{\begin{array}{ll}
\frac {-1+(-1)^{sm-1} \sqrt r} 2, \mbox{ if } p \equiv 1 \pmod 4, \\
\frac {-1+(-1)^{sm-1} (\sqrt {-1})^{sm} \sqrt r} 2, \mbox{ if }  p \equiv 3 \pmod 4, \end{array}\right.$$
and $\eta_1^{(2, r)}=-1-\eta_0^{(2, r)}$.
\end{lem}
The Gauss periods in the semi-primitive case are known and are described in the following lemma.

\begin{lem} \cite{My} Assume that there exists a least positive integer $e$ such that
$p^e \equiv -1\pmod N$. Let $r=p^{2ef}$ for some positive integer $f$.
\begin{enumerate}
  \item If $f, p$, and $\frac {p^e+1} N$ are all odd, then
  $$\eta_{N/2}^{(N, r)}=\frac {(N-1)\sqrt r-1} N, \eta_i^{(N, r)}=-\frac {\sqrt r+1} N \mbox { for }
  i \neq N/2.$$
  \item In all other cases,
  $$\eta_0^{(N, r)}=\frac {(-1)^{f+1}(N-1)\sqrt r-1} N, \eta_i^{(N, r)}=\frac {(-1)^f\sqrt r-1} N
  \mbox { for } i \neq 0.$$
\end{enumerate}
\end{lem}

The Gauss periods in the index $2$ case can be described in the following lemma.

\begin{lem} \cite{DY, My} Let $N > 3$  be a prime with $N \equiv 3 \pmod 4$, $p$ a prime such
that $[\Bbb Z_N^\ast : \langle p \rangle]=2$, and $r=p^{\frac {N-1} 2 k}$ for some positive integer $k$.
Let $h$ be the class number of $\Bbb Q(\sqrt {-N})$ and $a,b$ the integers satisfying
$$\left\{\begin{array}{lll}
4p^h=a^2+Nb^2 \\
a \equiv -2p^{\frac {N-1+2h}4} \pmod N \\
b>0, p \nmid b. \end{array}\right.$$
Then the Gauss periods of order $N$ are given by
$$\left\{\begin{array}{lll}
\eta_0^{(N, r)}=\frac 1 N (P^{(k)}A^{(k)}(N-1)-1) \\
\eta_u^{(N,r)}=\frac {-1} N(P^{(k)}A^{(k)}+P^{(k)}B^{(k)}N+1), \mbox{ if } (\frac u N)=1 \\
 \eta_u^{(N,r)}=\frac {-1} N(P^{(k)}A^{(k)}-P^{(k)}B^{(k)}N+1), \mbox{ if } (\frac u N)=-1, \end{array}\right.$$
 where
 $$\left\{\begin{array}{lll}
P^{(k)}=(-1)^{k-1}p^{\frac k 4(N-1-2h)} \\
A^{(k)}=\mbox{Re}(\frac {a+b\sqrt{-N}} 2)^k \\
B^{(k)}=\mbox{Im}(\frac {a+b\sqrt{-N}} 2)^k/ \sqrt N. \end{array}\right.$$
\end{lem}

\section{ Weights of $\mathcal C_{(q, m, n_1,n_2, \ldots, n_u)}$}

In this section, we  present a method to compute the weights of codewords in cyclic code
$\mathcal C_{(q, m, n_1,n_2, \ldots, n_u)}$ as (1.1) if $\gcd(n_i, n_j)=1$ for $1 \leq i \neq j \leq u$.

For any $a_1, a_2, \ldots, a_u \in \Bbb F_r$, the Hamming weight
of $c(a_1, a_2, \ldots, a_u)$ is equal to $$W_H(c(a_1, a_2, \ldots, a_u))= n-Z(r, a_1, a_2, \ldots, a_u),$$ where
$$Z(r, a_1, a_2, \ldots, a_u)=|\{i : \mbox{Tr}_{r/q}(a_1g_1^i+a_2g_2^i+\cdots +a_ug_u^i)=0, 0 \leq i \leq n-1\}|.$$

\begin{thm} Let each $g_i$ be of order $n_i$, $r-1=n_1N_1=\cdots=n_uN_u$ and $\gcd(n_i,n_j)=1$ for $1\le i\ne j\le u$, then the weight of the codeword $c(a_1, a_2, \ldots, a_u)$ as (1.2) is
$$\frac {(q-1)n} q-\frac 1 q \sum_{y \in \Bbb F_q^\ast}
\sum_{x_1 \in C_0^{(N_1, r)}}\psi(y a_1x_1) \cdots
\sum_{x_u \in C_0^{(N_u, r)}}\psi(y a_ux_u).$$
\end{thm}

 \begin{proof}
 Let $\phi$ be the canonical additive character of $\Bbb F_q$. Then $\psi=\phi \circ \mbox{Tr}_{r/q}$
is the canonical additive character of $\Bbb F_r$. By the orthogonal property of additive characters we have
\begin{eqnarray*} Z(r,a_1, a_2, \ldots, a_u)&=&\sum_{i=0}^{n-1} \frac  1 q \sum_{y \in \Bbb F_q}
\phi(y \mbox{Tr}_{r/q}(a_1g_1^i+a_2g_2^i+\cdots +a_ug_u^i)) \\
&=& \frac n q +\frac 1 q \sum_{y \in \Bbb F_q^\ast}\sum_{i=0}^{n-1}\phi (\mbox{Tr}_{r/q}(ya_1g_1^i+ya_2g_2^i+\cdots +ya_ug_u^i)) \\
&=& \frac n q +\frac 1 q \sum_{y \in \Bbb F_q^\ast}\sum_{i=0}^{n-1}\psi(ya_1g_1^i+ya_2g_2^i+\cdots +ya_ug_u^i) \\
&=& \frac n q +\frac 1 q \sum_{y \in \Bbb F_q^\ast}\sum_{i=0}^{n-1}\psi( ya_1g_1^i)\psi(ya_2g_2^i+\cdots +ya_{u}g_{u}^i).
\end{eqnarray*}

Denote $n_1'=n_2 \cdots n_u$. For $0 \leq i \leq n-1$, we have
$$i=sn_1'+t, 0 \leq s \leq n_1-1, 0 \leq t \leq n_1'-1.$$
 Then
\begin{eqnarray*} Z(r,a_1, a_2, \ldots, a_u) &=&
\frac n q +\frac 1 q \sum_{y \in \Bbb F_q^\ast}
\sum_{s=0}^{n_1-1}\sum_{t=0}^{n_1'-1}\psi(ya_1g_1^{sn_1'+t})\psi(ya_2g_2^t+\cdots +ya_ug_u^t))\\
&=& \frac n q +\frac 1 q \sum_{y \in \Bbb F_q^\ast}
\sum_{t=0}^{n_1'-1}\psi(ya_2g_2^t+\cdots +ya_ug_u^t)\sum_{s=0}^{n_1-1}\psi(y a_1g_1^{sn_1'+t}).
\end{eqnarray*}

Since $\gcd(n_1, n_1')=1$, for any fixed $t$ we have the inner sum $$\sum_{s=0}^{n_1-1}\psi(y a_1g_1^{sn_1'+t})=\sum_{s=0}^{n_1-1}\psi(y a_1g_1^s).$$
 Thus we have
\begin{eqnarray*} Z(r,a_1, a_2, \ldots, a_u)&=& \frac n q +\frac 1 q \sum_{y \in \Bbb F_q^\ast}
\sum_{s=0}^{n_1-1}\psi(y a_1g_1^s)\sum_{t=0}^{n_1'-1}\psi(ya_2g_2^t+\cdots +ya_ug_u^t) \\
&=& \frac n q +\frac 1 q \sum_{y \in \Bbb F_q^\ast}
\sum_{x_1 \in C_0^{(N_1, r)}}\psi(y a_1x_1)\sum_{t=0}^{n_1'-1}\psi(ya_2g_2^t+\cdots +ya_ug_u^t).
\end{eqnarray*}
Similarly, we have
\begin{equation}Z(r,a_1, a_2, \ldots, a_u)=\frac n q+
\frac 1 q \sum_{y \in \Bbb F_q^\ast}
\sum_{x_1 \in C_0^{(N_1, r)}}\psi(y a_1x_1) \cdots
\sum_{x_u \in C_0^{(N_u, r)}}\psi(y a_ux_u).\end{equation}
This completes the proof.
\end{proof}
\begin{cor} The assumptions are as Theorem 3.1.    If $q=2$, then the weight of codeword $c(a_1,\ldots, a_u)$ is $$\frac n2-\frac 12\overline{\eta}_{a_1}\cdots \overline{\eta}_{a_u},$$ where
$$\overline{\eta}_{a_i}=\left\{\begin{array}{ll} n_i,&\mbox{ if }a_i=0, \\ \eta_j^{(N_i, r)}, &\mbox{ if }a_i\in C_{j}^{(N_i, r)}\end{array}\right.$$
\end{cor}
\begin{proof} Note that $q=2$ and $\Bbb F_2^*=\{1\}$. Then by Theorem 3.1 we have
$$Z(r,a_1, a_2, \ldots, a_u)=\frac n 2+
\frac 1 2
\sum_{x_1 \in C_0^{(N_1, r)}}\psi( a_1x_1) \cdots
\sum_{x_u \in C_0^{(N_u, r)}}\psi(a_ux_u).$$
This completes the proof.
\end{proof}

 \section{Weight distribution of $\mathcal C_{(q, m, n_1, n_2)}$}

In this section, we shall determine the weight distribution of
$\mathcal C_{(q, m, n_1, n_2)}$ with $\gcd(n_1, n_2)=1$,
where $$\mathcal C_{(q, m, n_1, n_2)}=\{c(a,b) : a, b \in \Bbb F_r\}$$ and
$$c(a,b)=(\mbox{Tr}_{r/q}(ag_1^0+bg_2^0), \mbox{Tr}_{r/q}(ag_1^1+bg_2^1) \ldots, \mbox{Tr}_{r/q}(ag_1^{n_1n_2-1}+bg_2^{n_1n_2-1})).$$
It follows from Delsarte's Theorem \cite{D} that the code $\mathcal C_{(q, m, n_1,n_2)}$ is a cyclic code
over $\Bbb F_q$ with the parity-check polynomial $h_1(x)h_2(x)$, where $h_1(x)$ and
$h_2(x)$ are the minimal polynomial of $g_1^{-1}$ and $g_2^{-1}$ over $\Bbb F_q$, respectively.
Note that $c(a_1, b_1)=c(a_2, b_2)$ may appear even if
$(a_1, b_1)\neq (a_2, b_2)$. For example, if $g_1=1$, $a_1 \neq a_2$ and  $\mbox{Tr}_{r/q}(a_1)=\mbox{Tr}_{r/q}(a_2)$, then
$c(a_1, b)=c(a_2, b)$.

For any $a, b \in \Bbb F_r$, the Hamming weight
of $c(a, b)$ is equal to $$W_H(c(a, b))= n-Z(r, a, b),$$ where
$$Z(r, a, b)=|\{i : \mbox{Tr}_{r/q}(ag_1^i+bg_2^i)=0, 0 \leq i \leq n_1n_2-1\}|.$$

For future use, we introduce a lemma about Abelian group.
\begin{lem}  \cite{J} Let $H$ and $K$ be two subgroups of a finite Abelian group $G$.

(1) Then $h_1K=h_2K$  if and only if  $ h_1(H \cap K)=h_2(H \cap K)$
for $h_1, h_2 \in H$. Moreover, there is an isomorphism of  groups:
$HK/K \cong H/(H \cap K)$  and  $[HK : K]=[H :(H \cap K)],$
where $HK=\{hk : h \in H, k \in K\}.$

(2)
Then  there is an isomorphism of groups: $(H \times K)/\Delta \cong HK,$
where $H \times K=\{(h, k) : h \in H, k \in K\}$ and
$\Delta=\{(h, h^{-1}) : h \in H \cap K\}$.
Moreover,
$|H|\cdot|K|=|HK|\cdot|H\cap K|.$
\end{lem}

\begin{thm}
 If $\gcd(n_1, n_2)=1$, $N_0=\frac {r-1} {q-1}, d_1=\gcd(N_0, N_1), d_2=\gcd(N_0, N_2)$, and
 $d=\gcd(\frac {N_0N_2} {d_2}, N_1)$, then the weight distribution of the cyclic code $\mathcal C_{(q, m, n_1, n_2)}$ is given by Table 1.
\end{thm}

\[ \begin{tabular} {c} Table 1. Weight distribution of $\mathcal C_{(q, m,n_1, n_2)}$ when $\gcd(n_1, n_2)=1$ \\
\begin{tabular}{|c|c|}
  \hline
 Weight & Frequency \\
  \hline
        0   &    1\\
  $\frac {(q-1)n_1n_2} q -\frac {(q-1)d_1n_2} {qN_1} \eta_j^{(d_1, r)}$  & $\frac {r-1} {d_1}(0\leq j \leq d_1-1)$ \\
  $\frac {(q-1)n_1n_2} q -\frac {(q-1)d_2n_1} {qN_2} \eta_j^{(d_2, r)}$  & $\frac {r-1} {d_2}(0\leq j \leq d_2-1)$ \\
  $\frac {(q-1)n_1n_2} q -\frac {(q-1)dd_2} {qN_1N_2}
  \sum\limits_{i=0}^{\frac {N_2} {d_2}-1} \eta_{N_0i+j}^{(N_2, r)} \eta_{N_0i+k}^{(d, r)}$
  & $\frac {(r-1)^2} {dN_2}(\begin{array}{c} 0\leq j \leq N_2-1 \\0\leq k \leq d-1\end{array})$ \\
  \hline
\end{tabular}
\end{tabular}
\]

\begin{proof} By Theorem 3.1 we have
\begin{eqnarray*} Z(r,a,b)=\frac {n_1n_2} q +\frac 1 q \sum_{y \in \Bbb F_q^\ast}\sum_{x \in C_0^{(N_1, r)}}\psi(y ax)
\sum_{z \in C_0^{(N_2, r)}}\psi(y b z).
\end{eqnarray*}

We will compute the values of $Z(r, a, b)$ in the following four cases.

(1) If $a=0$ and $b=0$, then
$$Z(r, a, b)=n_1n_2.$$
This value occurs only once.

(2) If $a \neq 0$ and $b=0$. Suppose that $a\in C_j^{(d_1, r)}$,  $0\le j\le d_1-1$, where $d_1=\gcd(N_0, N_1)$.  Then
\begin{eqnarray*} Z(r,a,b)= \frac {n_1n_2} q +\frac {n_2} q \sum_{y \in \Bbb F_q^\ast}\sum_{x \in C_0^{(N_1, r)}}\psi(y ax).
\end{eqnarray*}
For each  $yx\in \Bbb F_q^*\cdot C_0^{(N_1, r)}=C_0^{(N_0,r)}\cdot C_0^{(N_1,r)}=C_0^{(d_1,r)}$, by Lemma 4.1 there exist exactly $\frac {(q-1)d_1} {N_1}$ elements  $h\in  \Bbb F_q^*\cap C_0^{(N_1, r)}$ such that $yh\cdot h^{-1}x=yx$.
Hence
we have
\begin{eqnarray*}
Z(r, a, b)&=&\frac {n_1n_2} q +\frac {n_2} q \cdot \frac {(q-1)d_1} {N_1}\sum_{x \in C_0^{(d_1, r)}} \psi(ax) \\
&=& \frac {n_1n_2} q +\frac {(q-1)d_1n_2} {qN_1} \eta_j^{(d_1, r)}.
\end{eqnarray*}
This value occurs $\frac {r-1} {d_1}$ times.

(3) If $a=0$ and $b \neq 0$. Suppose that $b\in C_j^{(d_2, r)}$,  $0\le j\le d_2-1$, where $d_2=\gcd(N_0, N_2)$.
Similarly, we have
\begin{eqnarray*} Z(r,a,b)&= &\frac {n_1n_2} q +\frac {n_1} q \sum_{y \in \Bbb F_q^\ast}\sum_{z \in C_0^{(N_2, r)}}\psi(y b z).
\\ &=&\frac {n_1n_2} q +\frac {n_1} q \cdot \frac {(q-1)d_2} {N_2}\sum_{z \in C_0^{(d_2, r)}} \psi(b z) \\
&=& \frac {n_1n_2} q +\frac {(q-1)d_2n_1} {qN_2} \eta_j^{(d_2, r)}.
\end{eqnarray*}
This value occurs $\frac {r-1} {d_2}$ times.

(4) If $a \neq 0, b \neq 0$. Suppose that  $a \in C_k^{(d, r)}$ and $b \in C_{j}^{(N_2, r)}$, $0\le k\le d-1$, $0\le j\le N_2-1$, where
$d=\gcd(\frac{N_0N_2}{d_2}, N_1)$.

Set $H=\Bbb F_q^\ast=C_0^{(N_0,r)}=\langle\alpha^{N_0}\rangle$ and $K=C_0^{(N_2, r)}=\langle \alpha^{N_2} \rangle$, where $\alpha$ is a primitive element of $\Bbb F_r$ and $N_0=\frac {r-1} {q-1}$. Then
$$HK=\langle \alpha^{d_2} \rangle, H \cap K=\langle \alpha^{\frac {N_0N_2} {d_2}} \rangle,$$
where $d_2=\gcd(N_0, N_2)$. By Lemma 4.1 we have
$[HK : K]=\frac {N_2} {d_2}$, $\frac{N_2}{d_2} \mid (q-1)$ and
$$\Bbb F_q^\ast=C_0^{(N_0, r)}=\cup_{i=0}^{\frac {N_2} {d_2}-1}C_{N_0i}^{( \frac {N_0N_2} {d_2}, r)},$$ where
$$C_{N_0 i}^{(\frac {N_0N_2} {d_2}, r)}
=\alpha^{N_0i}\langle \alpha^{\frac {N_0N_2} {d_2} } \rangle
=\alpha^{N_0i}(H \cap K)
, 0 \leq i \leq \frac {N_2} {d_2}-1.$$
Then by Lemma 4.1  we have
$$HK=\Bbb F_q^\ast \cdot C_0^{(N_2, r)}=\cup_{i=0}^{\frac {N_2} {d_2}-1}\alpha^{N_0i}C_0^{(N_2, r)}
=\cup_{i=0}^{\frac {N_2} {d_2}-1}C_{N_0i}^{(N_2, r)},$$ where
$C_{N_0i}^{(N_2, r)}=C_{N_0i \pmod {N_2}}^{(N_2, r)}$ .
We have
\begin{eqnarray} \Omega \notag &=&\sum\limits_{y \in \Bbb F_q^\ast}
\sum\limits_{x \in C_0^{(N_1, r)}}\psi(y ax)\sum\limits_{z \in C_0^{(N_2, r)}}\psi(y b z)\\ \notag
&=& \sum\limits_{i=0}^{\frac {N_2} {d_2}-1}
\sum\limits_{y \in C_{N_0i}^{(\frac {N_0N_2} {d_2}, r)}}
(\sum\limits_{x \in C_0^{(N_1, r)}}\psi(y ax))(\sum\limits_{z \in C_0^{(N_2, r)}}\psi(y b z))\\ \notag
&=& \sum\limits_{i=0}^{\frac {N_2} {d_2}-1} \sum\limits_{y \in C_{N_0i}^{(\frac {N_0N_2} {d_2}, r)}}
(\sum\limits_{x \in C_0^{(N_1, r)}}\psi(y ax))(\sum\limits_{z \in C_{N_0i}^{(N_2, r)}}\psi( b z)) \\
&=& \sum\limits_{i=0}^{\frac {N_2} {d_2}-1}\eta_{N_0i+j}^{(N_2, r)}
\sum\limits_{y \in C_{N_0i}^{(\frac {N_0N_2} {d_2}, r)}}\sum\limits_{x \in C_0^{(N_1, r)}}\psi(y ax).
\end{eqnarray}
It is easily known that $C_{N_0i}^{(\frac {N_0N_2} {d_2}, r)} \cdot C_0^{(N_1, r)}=C_{N_0i}^{(d,r)}$.
For each $yx \in C_{N_0i}^{(d,r)}$, by Lemma 4.1 there exist exactly $\frac {(q-1)dd_2} {N_1N_2}$ elements
$h \in C_0^{(\frac {N_0N_2}{d_2}, r)} \cap C_0^{(N_1, r)}$ such that $yh \cdot  h^{-1}x=yx$. Then we have

$$\Omega= \sum_{i=0}^{\frac {N_2} {d_2}-1}\eta_{N_0i+j}^{(N_2, r)}
(\frac {(q-1)dd_2} {N_1N_2} \sum_{z \in C_{N_0i}^{(d, r)}} \psi(az)) \\
= \frac {(q-1)dd_2} {N_1N_2}\sum_{i=0}^{\frac {N_2} {d_2}-1}\eta_{N_0i+j}^{(N_2, r)}
\eta_{N_0i+k}^{(d, r)}.$$

Hence
$$Z(r, a, b)=\frac {n_1n_2} q+\frac {(q-1)dd_2} {qN_1N_2}\sum_{i=0}^{\frac {N_2} {d_2}-1}\eta_{N_0i+j}^{(N_2, r)}
\eta_{N_0i+k}^{(d, r)}.$$ This value occurs $\frac {(r-1)^2}{dN_2}$ times.

 Note that $W_H(c(a, b))= n-Z(r, a, b)$. Then we can obtain the Table 1 and this completes the proof.
\end{proof}

\begin{exa}
Let $(q, m, n_1, n_2)=(4, 2, 5, 3)$. Then we have
$r=16, N_1=3, N_2=5, d_1=1, d_2=5$, and $d=1$.
By Lemma 2.2 we have
$$\eta_0^{(5, 16)}=3, \eta_i^{(5, 16)}=-1 \mbox{ for } i=1, 2, 3, 4.$$
Note that $\eta_0^{(1, 16)}=-1$. Then by Table 1 we know that
$\mathcal C_{(q, m, n_1, n_2)}$ is a $[15, 3, 11]$ cyclic code
over $\Bbb F_4$ with the weight enumerator
$$1+45x^{11}+15x^{12}+3x^{15}.$$
\end{exa}

\begin{exa}
Let $(q, m, n_1, n_2)=(8, 2, 9, 7)$. Then we have
$r=64, N_1=7, N_2=9, d_1=1, d_2=9$, and $d=1$.
By Lemma 2.2 we have
$$\eta_0^{(9, 64)}=7, \eta_i^{(9, 64)}=-1 \mbox{ for } i=1, 2, \ldots, 8.$$
Note that $\eta_0^{(1, 16)}=-1$. Then by Table 1 we know that
$\mathcal C_{(q, m, n_1, n_2)}$ is a $[63, 3, 55]$ cyclic code
over $\Bbb F_8$ with the weight enumerator
$$1+441x^{55}+63x^{56}+7x^{63}.$$
\end{exa}

\begin{exa}
Let $(q, m, n_1, n_2)=(9, 2, 5, 16)$. Then we have
$r=81, N_1=16, N_2=5, d_1=2, d_2=5$, and $d=2$.
By Lemma 2.1 we have
$$\eta_0^{(2, 81)}=-5, \eta_1^{(2, 81)}=4.$$
By Lemma 2.2 we have
$$\eta_0^{(5, 81)}=7, \eta_i^{(5, 81)}=-2 \mbox{ for } i=1, 2, 3, 4.$$
Note that $\eta_0^{(1, 81)}=-1$. Then by Table 1 we know that
$\mathcal C_{(q, m, n_1, n_2)}$ is a $[80, 4, 40]$ cyclic code
over $\Bbb F_9$ with the weight enumerator
$$1+16x^{40}+40x^{64}+640x^{68}+2560x^{70}+2560x^{72}+640x^{75}+104x^{80}.$$
\end{exa}

\begin{exa}
Let $(q, m, n_1, n_2)=(7, 2, 3, 16)$. Then we have
$r=49, N_1=16, N_2=3, d_1=8, d_2=1$, and $d=8$.
By Lemma 2.2 we have
$$\eta_4^{(8, 49)}=6, \eta_i^{(8, 49)}=-1 \mbox{ for } i \neq 4.$$
Note that $\eta_0^{(1, 81)}=-1$ and $\eta_0^{(3, 49)}+\eta_1^{(3, 49)}+\eta_2^{(3, 49)}=-1$. Then by Table 1 we know that
$\mathcal C_{(q, m, n_1, n_2)}$ is a $[48, 3, 41]$ cyclic code
over $\Bbb F_7$ with the weight enumerator
$$1+288x^{41}+48x^{42}+6x^{48}.$$
\end{exa}

If one of $n_1$ and $n_2$ is equal to 1, without loss of generality we assume that $n_1=1$, then
we have $N_1=r-1, d_1=N_0$, and $d=\frac {N_0N_2} {d_2}$ since
$\frac {N_2} {d_2} \mid (q-1)$. Thus we can get the weight distribution
of the cyclic code $\mathcal C_{(q, m, 1, n_2)}$. In fact, the first author and the second author \cite{LY}
had presented its weight distribution. Theorem 4.2 can be viewed as the generalization of the result
in \cite{LY}.
\begin{cor}
 Note that $c(a_1, b)=c(a_2, b)$ if $\mbox{Tr}_{r/q}(a_1)=\mbox{Tr}_{r/q}(a_2)$.
 Then the weight distribution of $\mathcal C_{(q, m, 1, n_2)}$ is given by Table 2.
\end{cor}

\[ \begin{tabular} {c} Table 2. Weight distribution of $\mathcal C_{(q, m, n_1, 1)}.$\\
\begin{tabular}{|c|c|}
  \hline
 Weight & Frequency  \\
  \hline
        0   &    1 \\
  $n_2$ & $q-1$ \\
  $\frac {(q-1)n_2} q -\frac {(q-1)d_2} {qN_2} \eta_j^{(d_2, r)}$ & $\frac {r-1} {d_2} (0 \leq j \leq d_2-1)$ \\
  $\frac {(q-1)n_2} q-\frac 1 q\sum\limits_{i=0}^{\frac {N_2} {d_2}-1}\eta_{N_0i+j}^{(N_2, r)}
\eta_{i+l}^{(\frac {N_2} {d_2}, q)}$  & $\frac {d_2(q-1)(r-1)} {N_2^2}
(\begin{array}{c}0 \leq j \leq N_2-1 \\  0 \leq l \leq \frac {N_2} {d_2}-1\end{array})$ \\
  \hline
\end{tabular}
\end{tabular}
\]

\begin{proof} For completeness we give a proof here. If $n_1=1$, then
we have $N_1=r-1, d_1=N_0$, and $d=\frac {N_0N_2} {d_2}$.
 By Theorem 3.1 we have
\begin{equation} Z(r,a,b)= \frac {n_2} q +\frac 1 q \sum_{y \in \Bbb F_q^\ast}\sum_{i=0}^{n_2-1}\psi(y a)\psi(y bg_2^i).\end{equation}

It is known that $\mbox{Tr}_{r/q}$ maps $\Bbb F_r$ onto $\Bbb F_q$ and
$$|\{ a \in \Bbb F_r : \mbox{Tr}_{r/q}(a)=\omega\}|=\frac r q$$ if fixing each $\omega \in \Bbb F_q$.
Note that $c(a_1, b)=c(a_2, b)$ if $\mbox{Tr}_{r/q}(a_1)=\mbox{Tr}_{r/q}(a_2)$. Then we can
determine the values of $Z(r, a, b)$ and their frequencies as follows.

(1) If $\mbox{Tr}_{r/q}(a)=0$ and $b=0$. We have
$$Z(r, a, b)=n_2.$$
This value occurs only once.

(2) If  $\mbox{Tr}_{r/q}(a)\ne 0$ and $b=0$, then
$$Z(r, a, b)=\frac {n_2} q +\frac {n_2} q \sum_{y \in \Bbb F_q^\ast} \psi(ya)
  = \frac {n_2} q +\sum_{y \in \Bbb F_q^\ast}\phi(y \mbox{Tr}_{r/q}(a))=0.$$
This value occurs $q-1$ times.

(3) If   $\mbox{Tr}_{r/q}(a)= 0$ and $b \neq 0$. By (4.2) and Lemma 4.1 we have
$$Z(r, a, b)=\frac {n_2} q +\frac {n_2} q \sum_{y \in \Bbb F_q^\ast} \phi(y \mbox{Tr}_{r/q}(a))\sum_{x\in C_0^{(N_2,r)}}\psi(ybx)= \frac {n_2} q +\frac {(q-1)d_2} {qN_2} \eta_j^{(d_2, r)}$$
for $b \in C_j^{(d_2, r)}$, where $d_2=\gcd(N_0,N_2)$. This value occurs $\frac {r-1} {d_2}$ times.

(4) If   $\mbox{Tr}_{r/q}(a)\ne 0$ and $b \neq 0$.
Suppose that $b \in C_{j}^{(N_2, r)}$ and  $\mbox{Tr}_{r/q}(a) \in C_k^{(\frac {N_2}{d_2}, q)}$, $0\le j\le N_2-1$, $0\le k\le \frac{N_2}{d_2}-1$, where $d_2=\gcd(N_0,N_2)$.
Then by (4.1) we have
$$Z(r, a, b)=\frac {n_2} q+\frac {1} {q}\sum_{i=0}^{\frac {N_2} {d_2}-1}\eta_{N_0i+j}^{(N_2, r)}
\sum_{y \in C_{N_0i}^{(\frac {N_0N_2} {d_2}, r)}}\psi(y a).$$
Moreover,
\begin{eqnarray*} \sum_{y \in C_{N_0i}^{(\frac {N_0N_2} {d_2}, r)}}\psi(y a)=\sum_{y\in C_{i}^{(\frac{N_2}{d_2}, q)}}\phi(y\mbox{Tr}_{r/q}(a))=\eta_{i+l}^{(\frac{N_2}{d_2}, q)}.
\end{eqnarray*}
This value occurs
$\frac {r-1} {N_2} \cdot \frac {(q-1)d_2} {N_2}=\frac {d_2(q-1)(r-1)}{N_2^2}$ times.

 Note that $W_H(c(a, b))= n-Z(r, a, b)$. Then we can obtain the Table 2 and this completes the proof.
\end{proof}

\section{ Weight distribution of $\mathcal C_{(q, m, n_1, n_2, 1)}$}

In this section, we shall determine the weight distribution of
$$\mathcal C_{(q, m, n_1,n_2, 1)}=\{ c(a, b, c) : a, b, c \in \Bbb F_r\},$$
 where $\gcd(n_1, n_2)=1$ and
 $$c(a, b)=(\mbox{Tr}_{r/q}(ag_1^0+bg_2^0+c), \mbox{Tr}_{r/q}(ag_1^1+bg_2^1+c), \ldots, \mbox{Tr}_{r/q}(ag_1^{n-1}+bg_2^{n-1}+c)).$$
 It follows from Delsarte's Theorem \cite{D} that the code $\mathcal C_{(q, m, n_1,n_2, 1)}$ is a cyclic code
over $\Bbb F_q$ with the parity-check polynomial $(x-1)h_1(x)h_2(x)$, where $h_1(x)$ and
$h_2(x)$ are the minimal polynomial of $g_1^{-1}$ and $g_2^{-1}$ over $\Bbb F_q$, respectively.
For any $a, b \in \Bbb F_r$, the Hamming weight
of $c(a, b,c)$ is equal to $$W_H(c(a, b, c))= n-Z(r, a, b, c),$$ where
$$Z(r, a, b, c)=|\{i : \mbox{Tr}_{r/q}(ag_1^i+bg_2^i+c)=0, 0 \leq i \leq n_1n_2-1\}|.$$

\begin{thm}
 If $\gcd(n_1, n_2)=1$, $N_0=\frac {r-1} {q-1}, d_1=\gcd(N_0, N_1), d_2=\gcd(N_0, N_2)$, and
 $d=\gcd(\frac {N_0N_2} {d_2}, N_1)$, then the weight distribution of the cyclic code $\mathcal C_{(q, m, n_1, n_2, 1)}$ is given by Table 3.
\end{thm}
\[ \begin{tabular} {c} Table 3. Weight distribution of $\mathcal C_{(q, m,n_1, n_2, 1)}$ when $\gcd(n_1, n_2)=1$ \\
\begin{tabular}{|c|c|}
  \hline
 Weight & Frequency \\
  \hline
        0   &    1\\
  $n_1n_2$   &   $q-1$ \\
  $\frac {(q-1)n_1n_2} q -\frac {(q-1)d_2n_1} {qN_2} \eta_j^{(d_2, r)}$  & $\frac {r-1} {d_2}(0\leq j \leq d_2-1)$ \\
  $\frac {(q-1)n_1n_2} q-\frac {n_1} q \sum\limits_{i=0}^{\frac {N_2} {d_2}-1}
\eta_{N_0i+j}^{(N_2,r)}\eta_{i+k}^{(\frac {N_2} {d_2}, q)}$ &
$\frac {d_2(q-1)(r-1)} {N_2^2}(\begin{array}{c} 0\leq j \leq N_2-1 \\0\leq k \leq \frac {N_2}{d_2}-1 \end{array})$  \\
  $\frac {(q-1)n_1n_2} q -\frac {(q-1)d_1n_2} {qN_1} \eta_j^{(d_1, r)}$  & $\frac {r-1} {d_1}(0\leq j \leq d_1-1)$ \\
  $\frac {(q-1)n_1n_2} q-\frac {n_2} q \sum\limits_{i=0}^{\frac {N_1} {d_1}-1}
\eta_{N_0i+j}^{(N_1,r)}\eta_{i+k}^{(\frac {N_1} {d_1}, q)})$ &
$\frac {d_1(q-1)(r-1)} {N_1^2}(\begin{array}{c} 0\leq j \leq N_1-1 \\0\leq k \leq \frac {N_1}{d_1}-1 \end{array})$  \\
  $\frac {(q-1)n_1n_2} q -\frac {(q-1)dd_2} {qN_1N_2}
  \sum\limits_{i=0}^{\frac {N_2} {d_2}-1} \eta_{N_0i+j}^{(N_2, r)} \eta_{N_0i+k}^{(d, r)}$
  & $\frac {(r-1)^2} {dN_2}(\begin{array}{c} 0\leq j \leq N_2-1 \\ 0\leq k \leq d-1\end{array})$ \\
  $\frac {(q-1)n_1n_2} q -\frac 1 q \sum\limits_{i=0}^{\frac {N_2} {d_2}-1} \eta_{N_0i+k}^{(N_2, r)}
\sum\limits_{i'=0}^{\frac {N_1} d-1}\eta_{N_0(i+\frac{N_2}{d_2}i')+j}^{(N_1, r)}
\eta_{i+\frac{N_2}{d_2}i'+l}^{(\frac {N_1N_2}{dd_2}, q)}$ &
$\frac {dd_2(q-1)(r-1)^2} {N_1^2N_2^2}(\begin{array}{c} 0\leq j \leq N_1-1 \\ 0\leq k \leq N_2-1 \\
0\leq l \leq \frac {N_1N_2}{dd_2}-1 \end{array})$ \\
  \hline
\end{tabular}
\end{tabular}
\]

\begin{proof} By Theorem 3.1 we have
$$Z(r,a_1, a_2, \ldots, a_u)=\frac n q+
\frac 1 q \sum_{y \in \Bbb F_q^\ast}\psi(yc)
\sum_{x_1 \in C_0^{(N_1, r)}}\psi(y ax_1)
\sum_{x_2 \in C_0^{(N_2, r)}}\psi(y bx_2).$$
It is known that $\mbox{Tr}_{r/q}$ maps $\Bbb F_r$ onto $\Bbb F_q$ and
$$|\{ c \in \Bbb F_r : \mbox{Tr}_{r/q}(c)=\omega\}|=\frac r q$$ for each $\omega \in \Bbb F_q$.
Note that $c(a, b, c_1)=c(a, b, c_2)$ if $\mbox{Tr}_{r/q}(c_1)=\mbox{Tr}_{r/q}(c_2)$. Then we can
determine the values of $Z(r, a, b, c)$ and their frequencies as follows.

(1) If $a=0, b=0$, and $\mbox{Tr}_{r/q}(c)=0$, then we have
$$Z(r,a,b,c)=n_1n_2.$$
This value occurs once.

(2) If $a=0, b=0$, and $\mbox{Tr}_{r/q}(c)\neq 0$, then we have
$$Z(r,a,b,c)=0.$$
This value occurs $q-1$ times.

(3) If $a=0, b \neq 0$, and $\mbox{Tr}_{r/q}(c)=0$. By Corollary 4.7 we have
$$Z(r,a,b,c)=n_1(\frac {n_2} q+\frac{(q-1)d_2} {qN_2} \eta_j^{(d_2, r)})$$
for $b \in C_j^{(d_2, r)}$. This value occurs $\frac {r-1} {d_2}$ times.

(4) If $a=0, b \neq 0$, and $\mbox{Tr}_{r/q}(c)\neq 0$. By Corollary 4.7 we have
$$Z(r,a,b,c)=n_1(\frac {n_2} q+\frac 1 q \sum_{i=0}^{\frac {N_2} {d_2}-1}
\eta_{N_0 i+j}^{(N_2,r)}\eta_{i+k}^{(\frac {N_2} {d_2}, q)})$$
for $b \in C_{j}^{(N_2, r)}$ and $\mbox{Tr}_{r/q}(c)\in C_{k}^{(\frac {N_2} {d_2}, q)}$.
This value occurs $\frac {d_2(q-1)(r-1)} {N_2^2}$ times.

(5) If $a \neq 0, b=0$, and $\mbox{Tr}_{r/q}(c)=0$. By Corollary 4.7 we have
$$Z(r,a,b,c)=n_2(\frac {n_1} q+\frac{(q-1)d_1} {qN_1} \eta_j^{(d_1, r)})$$
for $a \in C_j^{(d_1, r)}$. This value occurs $\frac {r-1} {d_1}$ times.

(6) If $a \neq 0, b=0$, and $\mbox{Tr}_{r/q}(c)\neq 0$. By Corollary 4.7 we have
$$Z(r,a,b,c)=n_2(\frac {n_1} q+\frac 1 q \sum_{i=0}^{\frac {N_1} {d_1}-1}
\eta_{N_0 i+j}^{(N_1,r)}\eta_{i+k}^{(\frac {N_1} {d_1}, q)})$$
for $a \in C_{j}^{(N_1, r)}$ and $\mbox{Tr}_{r/q}(c)\in C_{k}^{(\frac {N_1} {d_1}, q)}$.
This value occurs $\frac {d_1(q-1)(r-1)} {N_1^2}$ times.

(7) If $a \neq 0, b \neq 0$, and $\mbox{Tr}_{r/q}(c)=0$. By Theorem 4.2 we have
$$Z(r, a, b)=\frac {n_1n_2} q+\frac {(q-1)dd_2} {qN_1N_2}\sum_{i=0}^{\frac {N_2} {d_2}-1}
\eta_{N_0i+j}^{(N_2, r)}\eta_{N_0i+k}^{(d, r)}$$
for $a \in C_k^{(d, r)}$ and $b \in C_{j}^{(N_2, r)}$.
This value occurs $\frac {(r-1)^2}{dN_2}$ times.

(8) If $a \neq 0, b \neq 0$, and $\mbox{Tr}_{r/q}(c)\neq 0$. Suppose that $a \in C_{j}^{(N_1, r)}, b \in C_{k}^{(N_2, r)}$,
and $\mbox{Tr}_{r/q}(c) \in C_{l}^{(\frac {N_1N_2} {dd_2}, q)}$.
Note that $C_{N_0 i}^{(\frac {N_0N_2} {d_2}, r)}
= C_i^{(\frac {N_2} {d_2}, q)}$ and $\Bbb F_q^\ast=\cup_{i=0}^{\frac {N_2} {d_2}-1}C_i^{(\frac {N_2} {d_2}, q)}$.
Then by the proof of
 Theorem 3.1 and Lemma 4.1 we have
\begin{eqnarray*}
Z(r,a,b,c)
&=& \frac {n_1n_2} q +\frac 1 q \sum_{y \in \Bbb F_q^\ast}\psi(y c)\sum_{x \in C_0^{(N_1, r)}}\psi(yax)\sum_{z \in C_0^{(N_2, r)}}\psi(ybz) \\
&=& \frac {n_1n_2} q +\frac 1 q \sum_{i=0}^{\frac {N_2} {d_2}-1}\sum_{y \in C_i^{(\frac {N_2} {d_2}, q)}} \psi(yc)
\sum_{x \in C_0^{(N_1, r)}}\psi(yax)\sum_{z \in C_0^{(N_2, r)}}\psi(ybz) \\
&=& \frac {n_1n_2} q +\frac 1 q \sum_{i=0}^{\frac {N_2} {d_2}-1}\sum_{y \in C_i^{(\frac {N_2} {d_2}, q)}} \psi(yc)
\sum_{x \in C_0^{(N_1, r)}}\psi(yax)\sum_{z \in C_{N_0i}^{(N_2, r)}}\psi(bz) \\
&=& \frac {n_1n_2} q +\frac 1 q \sum_{i=0}^{\frac {N_2} {d_2}-1} \eta_{N_0i+k}^{(N_2, r)}
\sum_{y \in C_i^{(\frac {N_2} {d_2}, q)}} \psi(yc)\sum_{x \in C_0^{(N_1, r)}}\psi(yax).
\end{eqnarray*}

 Denote $H=C_0^{(\frac {N_2}{d_2}, q)}=C_0^{(\frac{N_0N_2}{d_2}, r)}$ and $K=C_0^{(N_1,r)}$. Then by Lemma 4.1 we know $H\cap K=\langle \alpha^{\frac{N_0N_1N_2}{dd_2}}\rangle$ and $H\cdot K=\langle \alpha^{d}\rangle$, where $d=\gcd(\frac{N_0N_2}{d_2}, N_1)$.  Hence we have
$$C_i^{(\frac {N_2}{d_2}, q)}=\alpha^{N_0i} \langle \alpha^{\frac {N_0N_2}{d_2}} \rangle
=\alpha^{N_0i} \cup_{i'=0}^{\frac {N_1} d-1}\alpha^{ \frac {N_0N_2}{d_2}i'}
\langle \alpha^{\frac {N_0N_2}{d_2} \cdot \frac {N_1} d} \rangle$$
and
\begin{eqnarray*} C_i^{(\frac {N_2}{d_2}, q)} \cdot C_0^{(N_1, r)}&=&\alpha^{N_0i} \cdot
\cup_{i'=0}^{\frac {N_1} d-1}\alpha^{\frac {N_0N_2}{d_2}i'} C_0^{(N_1, r)} \\
&=& \cup_{i'=0}^{\frac {N_1} d-1}\alpha^{N_0 (i+\frac {N_2}{d_2}i')} C_0^{(N_1, r)} \\
&=& \cup_{i'=0}^{\frac {N_1} d-1}C_{N_0 (i+\frac {N_2}{d_2}i')}^{(N_1, r)}.
\end{eqnarray*}
Therefore we have
\begin{eqnarray*}
Z(r,a,b,c)&=& \frac {n_1n_2} q +\frac 1 q \sum_{i=0}^{\frac {N_2} {d_2}-1} \eta_{N_0i+k}^{(N_2, r)}
\sum_{i'=0}^{\frac {N_1} d-1} \sum_{y \in C_{i+\frac {N_2}{d_2}i'}^{(\frac {N_1N_2}{dd_2}, q)}} \psi(yc)
\sum_{x \in C_0^{(N_1, r)}}\psi(yax) \\
&=& \frac {n_1n_2} q +\frac 1 q \sum_{i=0}^{\frac {N_2} {d_2}-1} \eta_{N_0i+k}^{(N_2, r)}
\sum_{i'=0}^{\frac {N_1} d-1} \sum_{y \in C_{i+\frac {N_2}{d_2}i'}^{(\frac {N_1N_2}{dd_2}, q)}} \psi(yc)
\sum_{x \in C_{N_0(i+\frac{N_2}{d_2}i')}^{(N_1, r)}}\psi(ax) \\
&=& \frac {n_1n_2} q +\frac 1 q \sum_{i=0}^{\frac {N_2} {d_2}-1} \eta_{N_0i+k}^{(N_2, r)}
\sum_{i'=0}^{\frac {N_1} d-1}\eta_{N_0(i+\frac{N_2}{d_2}i')+j}^{(N_1, r)}
\eta_{i+\frac{N_2}{d_2}i'+l}^{(\frac {N_1N_2}{dd_2}, q)}.
\end{eqnarray*}
This value occurs $\frac {dd_2(q-1)(r-1)^2} {N_1^2N_2^2}$ times.

 Note that $W_H(c(a, b, c))= n_1n_2-Z(r, a, b, c)$. Then we can obtain the Table 3 and this completes the proof.
\end{proof}

\begin{exa}
Let $(q, m, n_1, n_2)=(4, 2, 5, 3)$. Then we have
$r=16, N_1=3, N_2=5, d_1=1, d_2=5, d=1, \frac {N_1}{d_1}=3, \frac{N_2}{d_2}=3, \frac {N_1} d=3$, and
$\frac {N_1N_2} {dd_2}=3$.
By Lemma 2.2 we have
$$\eta_0^{(3, 4)}=1, \eta_1^{(3, 4)}=\eta_2^{(3, 4)}=-1,$$
$$\eta_0^{(3, 16)}=-3, \eta_1^{(3, 16)}=\eta_2^{(3, 16)}=1,$$
$$\eta_0^{(5, 16)}=3, \eta_i^{(5, 16)}=-1 \mbox{ for } i=1, 2, 3, 4.$$
Note that $\eta_0^{(1, 16)}=\eta_0^{(1, 4)}=-1$. Then by Table 3 we know that
$\mathcal C_{(q, m, n_1, n_2, 1)}$ is a $[15, 4, 9]$ cyclic code
over $\Bbb F_4$ with the weight enumerator
$$1+30x^9+54x^{10}+45x^{11}+105x^{12}+21x^{15}.$$
\end{exa}

\begin{exa}
Let $(q, m, n_1, n_2)=(8, 2, 9, 7)$. Then we have
$r=64, N_1=7, N_2=9, d_1=1, d_2=9, d=1, \frac {N_1} {d_1}=7, \frac {N_2} {d_2}=1, \frac {N_1} d=7$, and
$\frac {N_1N_2} {dd_2}=7$.
By Lemma 2.2 we have
$$\eta_0^{(9, 64)}=7, \eta_i^{(9, 64)}=-1 \mbox{ for } i=1, 2, \ldots, 8.$$
By Lemma 2.3 we have
$$\eta_0^{(7, 8)}=-1, \eta_1^{(7, 8)}=\eta_2^{(7, 8)}=\eta_4^{(7, 8)}=-1, \eta_3^{(7, 8)}=\eta_5^{(7, 8)}=\eta_6^{(7, 8)}=1$$
and
$$\eta_0^{(7, 64)}=5, \eta_1^{(7, 64)}=\eta_2^{(7, 64)}=\eta_4^{(7, 64)}=-3, \eta_3^{(7, 64)}=\eta_5^{(7, 64)}=\eta_6^{(7, 64)}=1.$$
Note that $\eta_0^{(1, 8)}=\eta_0^{(1, 64)}=-1$. Then by Table 3 we know that
$\mathcal C_{(q, m, n_1, n_2)}$ is a $[63, 4, 49]$ cyclic code
over $\Bbb F_8$ with the weight enumerator
$$1+252x^{49}+1372x^{54}+441x^{55}+1827x^{56}+203x^{63}.$$
\end{exa}

\begin{thm}
 If $q=2$ and $\gcd(n_1, n_2)=1$,
 then the weight distribution of the cyclic code $\mathcal C_{(q, m, n_1, n_2, 1)}$ is given by Table 4.
\end{thm}

\[ \begin{tabular} {c} Table 4. Weight distribution of $\mathcal C_{(q, m,n_1, n_2, 1)}$ when $N_1 \mid N_0$ and $N_2 \mid N_0$. \\
\begin{tabular}{|c|c|}
  \hline
 Weight & Frequency \\
  \hline
        0   &    1\\
  $n_1n_2$   &   1 \\
  $\frac {n_1n_2} 2 -\frac {n_1} 2 \eta_k^{(N_2, r)}$  & $\frac {r-1} {N_2}(0\leq k \leq N_2-1)$ \\
  $\frac {n_1n_2} 2 +\frac {n_1} 2 \eta_k^{(N_2, r)}$  & $\frac {r-1} {N_2}(0\leq k \leq N_2-1)$ \\
  $\frac {n_1n_2} 2 -\frac {n_2} 2 \eta_j^{(N_1, r)}$  & $\frac {r-1} {N_1}(0\leq j \leq N_1-1)$ \\
  $\frac {n_1n_2} 2 +\frac {n_2} 2 \eta_j^{(N_1, r)}$  & $\frac {r-1} {N_1}(0\leq j \leq N_1-1)$ \\
  $\frac {n_1n_2} 2 -\frac 1 2 \eta_j^{(N_1, r)}\eta_k^{(N_2, r)}$  & $\frac {r-1} {N_2}
  (\begin{array}{c} 0\leq j \leq N_1-1 \\ 0\leq k \leq N_2-1\end{array})$ \\
  $\frac {n_1n_2} 2 +\frac 1 2 \eta_j^{(N_1, r)}\eta_k^{(N_2, r)}$  & $\frac {r-1} {N_2}
  (\begin{array}{c} 0\leq j \leq N_1-1 \\ 0\leq k \leq N_2-1\end{array})$ \\
  \hline
\end{tabular}
\end{tabular}
\]

\begin{proof}
If $q=2$, then $N_0=r-1$, $N_1 \mid N_0$, and $N_2 \mid N_0$. Thus we have
$d_1=N_1, d_2=N_2$, and $d=N_1$.
Then Table 4 can be obtained by Table 3 and this completes the proof.
\end{proof}

\begin{exa}
Let $(q, m, n_1, n_2)=(2, 4, 5, 3)$. Then we have
$r=16, N_1=3, N_2=5$.
By Lemma 2.2 we have
$$\eta_0^{(3, 16)}=-3, \eta_1^{(3, 16)}=\eta_2^{(3, 16)}=1,$$
$$\eta_0^{(5, 16)}=3, \eta_i^{(5, 16)}=-1 \mbox{ for } i=1, 2, 3, 4.$$
 Then by Table 4 we know that
$\mathcal C_{(q, m, n_1, n_2, 1)}$ is a $[15, 7, 3]$ cyclic code
over $\Bbb F_2$ with the weight enumerator
$$1+5x^3+3x^5+25x^6+30x^7+30x^8+25x^9+3x^{10}+5x^{12}+x^{15}.$$
\end{exa}

\section{Concluding remarks}

In this paper, we have presented a method to compute the weights of codewords of cyclic code
$\mathcal C_{(q, m, n_1,n_2, \ldots, n_u)}$ and have determined the weight distributions of
the cyclic codes $\mathcal C_{(q, m, n_1,n_2)}$ and $\mathcal C_{(q, m, n_1,n_2,1)}$.

Let $g_1$ and $g_2$ be two elements of $\Bbb F_r^\ast$ with the coprime orders $n_1$ and
$n_2$, respectively. For a divisor $e_1$ of $n_1$, we assume that $g_1$ and $g_3=\mu_1 g_1$ are not conjugates of each other
 over $\Bbb F_q$ and $\ord(g_3)=n_1$, where $\mu_1 \in \Bbb F_r^\ast$ and $\ord(\mu)=e_1$. The weight distribution of
 the cyclic code $\mathcal C_{(q, m, n_1, n_2, n_1)}$ can be determined by our method and the method in \cite{DMLZ} or \cite{Ma}.
 Moreover, let $\mu_2$ be a elements of $\Bbb F_r^\ast$ and $\ord(\mu_2)=e_2$, where $e_2$ is a divisor of $n_2$.
 Suppose that $g_2$ and $g_4=\mu_2 g_2$ are not conjugates of each other over $\Bbb F_q$. Then the weight distribution
 of the cyclic code $\mathcal C_{(q, m, n_1, n_2, n_1, n_2)}$ can also be determined by our method and the method in \cite{DMLZ} or \cite{Ma}.
 We leave this for future work.

\end{document}